\documentclass[runningheads]{llncs}

\usepackage{algorithm}
\usepackage{algpseudocode}
\usepackage{amsmath}
\usepackage{amssymb}
\usepackage{color}
\usepackage{csquotes}
\usepackage{graphicx}
\usepackage{pifont}
\usepackage{wrapfig}

\newcommand{\myAnd}{\ensuremath{\wedge}}

\newcommand{\A}{\ensuremath{\mathcal{A}}}

\newcommand{\N}{\ensuremath{\mathcal{N}}}

\newcommand{\T}{\ensuremath{\mathcal{T}}}

\newcommand{\init}{\ensuremath{\mathit{init}}}
\newcommand{\op}{\ensuremath{\mathit{op}}}
\newcommand{\linit}{\ensuremath{L_{\mathit{init}}}}
\newcommand{\execute}{\ensuremath{\mathit{execute}}}
\newcommand{\event}{\ensuremath{\mathit{event}}}

\newcommand{\size}{\ensuremath{\mathit{Usize}}}
\newcommand{\abssize}{\ensuremath{\mathit{absUsize}}}

\newcommand{\muengen}{\ensuremath{\mathit{Muen}}}
\newcommand{\setgen}{\ensuremath{\mathit{SetGen}}}

\newcommand{\subjectspecs} {\ensuremath{\mathit{subject\textrm{-}specs}}}
\newcommand{\subjectdescs} {\ensuremath{\mathit{subject\textrm{-}descriptors}}}
\newcommand{\schedplans} {\ensuremath{\mathit{scheduling\textrm{-}plans}}}
\newcommand{\IOBitmap} {\ensuremath{\mathit{IOBitmap}}}
\newcommand{\curmajfr} {\ensuremath{\mathit{cur\textrm{-}maj\textrm{-}frame}}}
\newcommand{\curminfr} {\ensuremath{\mathit{cur\textrm{-}min\textrm{-}frame}}}
\newcommand{\globevents} {\ensuremath{\mathit{global\textrm{-}events}}}
\newcommand{\pendingevents} {\ensuremath{\mathit{pending\textrm{-}events}}}
\newcommand{\vectorrouting} {\ensuremath{\mathit{vector\textrm{-}routing}}}
\newcommand{\page} {\ensuremath{\mathit{page}}}

\newcommand{\chmem} {\ensuremath{\mathit{chmem}}}
\newcommand{\vmem} {\ensuremath{\mathit{VMem}}}
\newcommand{\cmap} {\ensuremath{\mathit{cmap}}}
\newcommand{\ch} {\ensuremath{\mathit{ch}}}
\newcommand{\enabled} {\ensuremath{\mathit{enabled}}}
\newcommand{\ticks} {\ensuremath{\mathit{ticks}}}
\newcommand{\minticks} {\ensuremath{\mathit{minticks}}}
\newcommand{\idealmajfp} {\ensuremath{\mathit{idealmajfp}}}
\newcommand{\minorfp} {\ensuremath{\mathit{minorfp}}}
\newcommand{\idealcycles} {\ensuremath{\mathit{idealcycles}}}
\newcommand{\majorfp} {\ensuremath{\mathit{majorfp}}}
\newcommand{\cycles} {\ensuremath{\mathit{cycles}}}
\newcommand{\last} {\ensuremath{\mathit{last}}}
\newcommand{\perms} {\ensuremath{\mathit{perms}}}
\newcommand{\ansubs} {\ensuremath{\mathit{ANSubs}}}

\newcommand{\pmem} {\ensuremath{\mathit{PMem}}}
\newcommand{\PT} {\ensuremath{\mathit{PT}}}
\newcommand{\pt} {\ensuremath{\mathit{ptab}}}

\newcommand{\vmcs} {\ensuremath{\mathit{vmcs}}}
\newcommand{\vmptr} {\ensuremath{\mathit{vmptr}}}
\newcommand{\eptp} {\ensuremath{\mathit{eptp}}}

\newcommand{\nsubs} {\ensuremath{\mathit{NSubs}}}
\newcommand{\ncpus} {\ensuremath{\mathit{NCPUs}}}

\newcommand{\LA} {\ensuremath{\mathit{LA}}}
\newcommand{\PA} {\ensuremath{\mathit{PA}}}
\newcommand{\psize} {\ensuremath{\mathit{size}}}
\newcommand{\PTBitArray} {PTBitArray}

\begin{document}
\title{Verification of a Generative Separation Kernel}
%
%
\author{Inzemamul Haque\inst{1}
\and
Deepak D'Souza\inst{1}
\and
Habeeb P\inst{1}
\and
Arnab Kundu\inst{2} 
\and
Ganesh Babu\inst{2}}
\authorrunning{I. Haque et al.}
%
\institute{Indian Institute of Science, Bangalore, India 
\email{\{inzemamul,deepakd,habeebp\}@iisc.ac.in} \and
Centre for Artificial Intelligence and Robotics, Bangalore, India
\email{\{arnab.kundu.in,ganeshbabu.vn\}@gmail.com}\\
}
\maketitle              
\begin{abstract}
We present a formal verification of the functional
correctness of the Muen Separation Kernel.
Muen is representative of the class of modern separation kernels
that leverage hardware virtualization support, and are
\emph{generative} in nature in that they generate a specialized kernel
for each system configuration. 
These features pose substantial challenges to existing verification techniques.
We propose a verification framework called conditional parametric
refinement which allows us to formally reason about generative systems.
We use this framework to carry out a conditional refinement-based proof of
correctness of the Muen kernel generator.
Our analysis of several system configurations shows that our
technique is effective in producing mechanized proofs of
correctness, and also in identifying issues that may compromise
the separation property.

\end{abstract}
\section{Introduction}
\label{sec:intro}

A separation kernel (SK) is a small specialized operating system or
microkernel, that provides a sand-boxed or ``separate''  execution
environment for a given set of processes (also called ``partitions''
or ``subjects'').
The subjects may communicate only via declared memory channels,
and are otherwise isolated from each other.
Unlike a general operating system 
these kernels usually have a
fixed set of subjects to run according to a specific schedule on
the different CPUs of a processor-based system.

The idea of a separation kernel was proposed by Rushby
\cite{rushby} as a way of breaking down the process of reasoning
about the overall security of a computer system.
The overall security of a system, in his view, derives from (a)
the physical separation of its components, and (b) the specific security
properties of its individual components or subjects.
A separation kernel thus focuses on providing the separation property
(a) above.
Separation kernels have since been embraced extensively in military
and aerospace domains, for building security and safety-critical
applications.

Our focus in this paper is in formally verifying that a separation
kernel does indeed provide the separation property, and more generally that
it functions \emph{correctly} (which would include for example, that
it executes subjects according to a specified schedule).
One way of obtaining a high level of assurance in the
correct functioning of a system is to carry out a refinement-based
proof of functional correctness \cite{Hoare75a,Hoare-tr-1985},
as has been done in the
context of OS verification \cite{rushby,klein}.
Here one specifies an abstract model of
the system's behaviour, and then shows that the system
implementation conforms to the abstract specification.
A refinement proof typically subsumes all the standard security
properties related to separation, like no-exfiltration/infiltration
and temporal and spatial separation of subjects considered for
instance in \cite{heitmeyer}.

Our specific aim in this paper is to formally verify the correctness of the
Muen separation kernel \cite{muen},
which is an open-source representative of a class of modern
separation kernels (including several commercial products like
GreenHills Integrity Multivisor \cite{GHIM},  LynxSecure \cite{Lynx}, PikeOS
\cite{pikeos}, VxWorks MILS platform \cite{Wind}, and XtratuM \cite{xtratum})
that use hardware virtualization support and
are \emph{generative} in nature.
By the latter we mean that these tools take an input specification
describing the subjects and the schedule of execution, and
generate a tailor-made processor-based system that includes subject
binaries, page tables, and a kernel that acts like a Virtual
Machine Monitor (VMM).

When we took up this verification task over three years ago, a
few challenges stood out.
How does one reason about a system whose kernel makes use of
virtualization features in the underlying hardware in addition to
Assembly and a high-level language like Ada?
Secondly, how does one reason about a complex 4-level paging structure and
the translation function it induces?
Finally, and most importantly, how do we reason about a generative
system to show that for \emph{every} possible input specification,
it produces a correct artifact?
A possible approach for the latter would be to verify the generator code, along
the lines of the CompCert project \cite{xavier}.
However with the generator code running close to 41K LOC, with little
or no compositional structure, and not being designed for verification,
this would be a formidable task.
One could alternatively perform translation validation
\cite{pnueli} of an input specification of interest, but this
would require manual effort from scratch each time.

We overcame the first challenge of virtualization by simply
choosing to model the virtualization layer (in this case Intel's
VT-x layer) along with the rest of the hardware components like
registers and memory, programmatically in software.
Thus we modeled VT-x components like the per-CPU VMX-Timer and
EPTP as 64-bit variables in Ada, and implicit structures like the
VMCS as a record with appropriate fields as specified by Intel
\cite{intel2018}.
Instructions like VMLAUNCH were then implemented as methods that
accessed these variables.
As it turned out, virtualization was more of a boon than a bane,
as it simplifies the kernel's (and hence the prover's) job of
preemption and context-switching.

We solved the third problem of generativeness (and coincidentally
the second problem of page tables too), by leveraging a key
feature of such systems: the kernel is essentially a
\emph{template} which is largely fixed, independent of the input
specification.
The kernel accesses variables which represent input-specific details
like subject details and the schedule, and these structures are
generated by Muen based on the given input specification.
The kernel can thus be viewed as a \emph{parametric} program, much like a
method that computes using its formal parameter variables.
In fact, taking a step back, the whole processor system
generated by Muen can be viewed as a parametric program with
parameter values like the schedule, subject details, page tables, and memory
elements being filled in by the generator based on the input
specification.

This view suggests a novel two-step technique for verifying generative
systems that can be represented as parametric programs.
We call this approach \emph{conditional parametric refinement}.
We first perform a general verification 
step (independent of the input spec) to verify that the parametric
program refines a parametric abstract specification,
\emph{assuming} certain natural conditions on the parameter values (for
example \emph{injectivity} of the page tables) that
are to be filled in.
This first step essentially tells us that for \emph{any} input
specification $P$, if the parameters generated by the system generator
satisfy the
assumed conditions, then the generated system is correct vis-a-vis the
abstract specification.
In the second step, which is \emph{input-specific}, we check
that for a given input specification, the
assumptions actually hold for the generated parameter values.
This gives us an effective verification technique for verifying
generative systems that lies
somewhere between verifying the generator and translation validation.

We carried out the first step of this proof technique for
Muen, using the Spark Ada \cite{gnatpro} verification
environment. The effort involved about 20K lines of source
code and annotation.
No major issues were found, modulo some subjective assumptions we
discuss in Sec.~\ref{sec:discussion}.
We have also implemented a tool that automatically and efficiently
performs the Step~2 check for a given SK configuration.
The tool is effective in proving the assumptions, leading to
machine-checked proofs of correctness for 12 different input
configurations, as well as in detecting issues 
like undeclared sharing of memory components in some seeded faulty
configurations.

To summarize our contributions, we have proposed a novel approach
based on parametric refinement for verifying generative systems.
We have applied this technique to verify the Muen separation kernel,
which is a complex low-level software system that makes use of hardware
virtualization features.
We believe that other verification efforts for similar generative
systems can benefit from our approach.
To the best of our knowledge our verification of Muen is the first
work that models and 
reasons about Intel's VT-x virtualization features in the context of
separation kernels.
Finally, we note that our verification of Muen is a \emph{post-facto}
effort, in that we verify an existing system which was not designed
and developed hand-in-hand with the verification process.

\section{Conditional Parametric Refinement} 
\label{sec:parametric-refinement}

We begin by introducing the flavour of classical refinement that
we will make use of, followed by the parametric refinement
framework 
we employ for our verification task.

\subsection{Machines and Refinement}
\label{sec:classical-refinement}
A convenient way to reason about systems such as Muen is to view
them as an \emph{Abstract Data Type} or simply \emph{machine} to
use the Event-B terminology \cite{abrial2010}.
A \emph{machine type} $\N = (N, \{I_n\}_{n \in N}, \{O_n\}_{n \in
  N})$ contains a finite set of named operations 
$N$, with each operation $n \in N$ having an
associated input domain $I_n$ and output domain $O_n$.
Each machine type contains a designated initialization
operation called \init.
A \emph{machine} $\A$ of type $\N$ above is a structure of the
form $(Q,\{\op_n\}_{n\in N})$, where $Q$ is a set of states, and for
each $n \in N$,
$\op_n : (Q \times I_n) \rightarrow (Q \times O_n)$
is the implementation of operation $n$.
If $\op_n(p,a) = (q,b)$, then when $\op_n$ is invoked with input $a$ in
a state $p$ of the machine, it returns $b$ and changes the state
to $q$.
The \init\ operation is assumed to be independent of the state in which it is
invoked as well as the input passed to it.
Hence we simply write $\init()$ instead of $\init(p,a)$ etc.

The machine $\A$ induces a transition system $\T_{\A}$ in a
natural way, whose states are the states $Q$ of $\A$, and
transitions from one state to another are labelled by triples of
the form $(n,a,b)$, representing that operation $n$ with input $a$ was
invoked and the return value was $b$.
One is interested in the language of \emph{initialized}
sequences of operation calls produced by this transition system, which
models behaviours of the system,
and we call it $\linit(\A)$.

We will consider different ways of representing machines, the most
important of which is as a program in a high-level imperative
programming language. 
The global variables of the
program (more precisely \emph{valuations} for them) make up the state
of the machine.
The implementation of an operation $n$ is given by a method definition
of the same name, that takes an argument $a$ in $I_n$, updates the
global state, and returns a value $b$ in $O_n$.
We call such a program a \emph{machine program}.
Fig~\ref{fig:program}(a) shows a program in a C-like
language, that represents a
``set'' machine with operations $\init$, $\mathit{add}$
and $\mathit{elem}$.
The set stores a subset of the numbers 0--3, in a Boolean array of size
4. However, for certain extraneous reasons, it uses an array $T$ to permute the
positions where information for an element $x$ is stored.
Thus to indicate that $x$ is present in the
set the bit $S[T[x]]$ is set to true.
We use the notation ``0..3'' to denote range of integers from 0
to 3 inclusive.

Another representation of a machine could be in the form of a
processor-based system. Here the state is given by the values of
the processor's registers and the contents of the memory. The
operations (like ``execute the next instruction on CPU0'', or ``timer
event on CPU1'') are defined by either the processor
hardware (as in the former operation) or by software in the form of an
interrupt handler (as in the latter operation).

\begin{figure}[th]
\centering
\begin{minipage}[t]{4.1cm}
\begin{scriptsize}
\begin{verbatim}
typedef univ 0..3;
bool S[4];
univ T[4] := {1,2,3,0};

void init() {
  for (int i:=0; i<4; i++)
    S[i] := false;
}
   
void add(univ x) {
  S[T[x]] := true;
}

bool elem(univ x) {
  return S[T[x]];
}
\end{verbatim}
\end{scriptsize}
\end{minipage}
\begin{minipage}[t]{3.6cm}
\begin{scriptsize}
\begin{verbatim}
// Abstract spec
typedef univ 0..3;
bool absS[4];
void add(univ x) {
  absS[x] := true;
}

...

// Gluing relation
\forall univ x:
   S[T[x]] = absS[x]
\end{verbatim}
\end{scriptsize}
\end{minipage}
\\
\hspace{-1.4cm}
\begin{minipage}[b]{4.1cm}
\begin{center}
\begin{small}
(a)
\end{small}
\end{center}
\end{minipage}
\begin{minipage}[b]{3.6cm}
\begin{center}
\begin{small}
(b)
\end{small}
\end{center}
\end{minipage}
\caption{(a) A machine program $P$ implementing a set machine and
  (b) an abstract specification $A$ and gluing relation.}
\label{fig:program}
\end{figure}

\begin{figure}
\centering
\begin{minipage}[t]{4.7cm}
\begin{scriptsize}
\begin{verbatim}
const unsigned Usize;
typedef univ 0..Usize-1;
bool S[Usize];
univ T[Usize];

void init() {
  for (int i:=0; i<Usize; i++)
    S[i] := false;
}
   
void add(univ x) {
  S[T[x]] := true;
}

bool elem(univ x) {
  return S[T[x]];
}
\end{verbatim}
\end{scriptsize}
\end{minipage}
\begin{minipage}[t]{4cm}
\begin{scriptsize}
\begin{verbatim}
// Abstract parametric spec
const unsigned absUsize;
typedef absUniv 0..absUsize-1
bool absS[absUsize];

void add(absUniv x) {
  absS[x] := true;
}

...

// Assumption R: Usize = absUsize 
     && T injective

// Parametric gluing relation
\forall univ x: 
   S[T[x]] = absS[x]
\end{verbatim}
\end{scriptsize}
\end{minipage} \\[0.3cm]
\hspace{-1.4cm}
\begin{minipage}[b]{4.7cm}
\begin{center}
\begin{small}
(a)
\end{small}
\end{center}
\end{minipage}
\begin{minipage}[b]{4cm}
\begin{center}
\begin{small}
(b)
\end{small}
\end{center}
\end{minipage}
\caption{(a) A parametric machine program $Q[\size,T]$ representing a parametric
      set machine, and 
  (b) abstract parametric specification $B[\abssize]$ and
      parametric gluing predicate.}
\label{fig:program-parametric}
\end{figure}

Refinement \cite{Hoare75a,Hoare-tr-1985,abrial2010} is a way of
saying that a ``concrete'' machine
conforms to an ``abstract'' one, behaviourally. 
In our setting
of total and deterministic machines,
refinement boils down to containment of sequential
behaviours.
Let $\A = (Q, \{\op_n\}_{n\in N})$  and $\A' = (Q', \{\op_n'\}_{n\in
  N})$ be two machines of type $\N$.
We say $\A'$ \emph{refines} $\A$ if $\linit(\A') \subseteq \linit(\A)$.
One usually makes use of a ``gluing'' relation to exhibit refinement.
A \emph{gluing relation} on the states of $\A'$ and $\A$ above is 
a relation $\rho \subseteq Q' \times Q$.
We say $\rho$ is \emph{adequate} (to show that $\A'$ refines $\A$)
if it satisfies the following conditions:
\label{sec:adequate}
\begin{itemize}
\item(init) Let $\op_{\init}() = q_0$ and $\op'_{\init}() =
  q_0'$.
  Then we require that $(q_0',q_0) \in \rho$.
  Thus, after the 
  machines are initialized, their states must be $\rho$-related.
\item(sim) Let $p \in Q$, and $p' \in Q'$, with $(p',p) \in
  \rho$. Let $n \in N$, $a \in I_n$, and suppose
  $\op_n(p,a) = (q, b)$ and $\op'_n(p',a) = (q', b')$. Then we must
  have $b=b'$ and $(q',q) \in \rho$.
\end{itemize}
It is not difficult to see that existence of an adequate gluing
invariant is sufficient for refinement.


When machines are presented in the form of programs, we can use
Floyd-Hoare logic based  code-level 
verification tools (like VCC \cite{VCC} for C, or GNAT Pro
\cite{gnatpro} for Ada Spark), to phrase the refinement conditions
as pre/post annotations and carry out a
machine-checked proof of refinement \cite{DivakaranDS14}.
The basic idea is to combine both the abstract and concrete programs
into a single ``combined'' program with separate state variables but
joint methods 
that carry out the abstract and concrete operation one after the
other.
The gluing relation is specified as a predicate on the combined state.
Fig.~\ref{fig:program}(b) shows an abstract specification
and a gluing relation, for the set machine program of part~(a).
The refinement conditions (init) and (sim) are phrased as pre/post
annotations on the joint operation methods, in the expected manner.

\subsection{Generative Systems and Parametric Refinement}
\label{sec:parametric-refinement-theory}

A \emph{generative system} is a program $G$ that given an input
specification $I$ (in some space of valid inputs), generates a machine
program $P_I$.
As an example, one can think of a set machine generator $\setgen$, that given a
number $n$ of type unsigned int (representing the universe size),
generates a program $P_n$ similar to
the one in Fig.~\ref{fig:program}(a), which uses the
constant $n$ in place of the set size 4, and an array $T_n$ of size
$n$, which maps each $x$ in $[0..n-1]$ to $(x+1)\!\!\!\mod n$.

For every $I$, let us say we have an abstract machine (again similar
to the one in Fig.~\ref{fig:program}(b)) say
$A_I$, describing the intended behaviour of the machine $P_I$.
Then the verification problem of interest to us, for the generative
system $G$, is to show that for \emph{each} input
specification $I$, $P_I$ refines $A_I$.
This is illustrated in Fig.~\ref{fig:generative-proof}(a).
We propose a way to address this problem using refinement of
\emph{parametric} programs, which we describe next.

\paragraph{Parametric Refinement.}

A \emph{parametric} program is like a standard program, except that
it has certain read-only variables which are left
\emph{uninitialized}. These uninitialized variables act like
``parameters'' to the program.
We denote by $P[V]$ a parametric program $P$ with an uninitialized
variable $V$.
As such a parametric program has no useful
meaning (since the uninitialized variables may contain
arbitrary values).
But if we initialize the variable $V$ with a value $v$ passed to the
program, we get a standard program which we denote by $P[v]$.
Thus $P[v]$ is obtained from $P[V]$ by adding the definition $V :=
v$ to the declaration of $V$.
By convention we use uppercase names for parameter variables, and
lowercase names for values passed to the program.
Instead of a single parameter we allow a parametric program to have a
list of parameters $\bar{V}$, and extend our notation in the expected
way for such programs.

Let $N$ be a set of operation names.
A \emph{parametric machine program} of type $N$ is a parametric
program $Q[\bar{V}]$ containing a method $f_n$ for each operation
$n \in N$.
The input/output types of $f_n$ may be dependent on and derived
from the parameter values.
Given a parameter value $\bar{v}$ for $\bar{V}$, we obtain the program
$Q[\bar{v}]$ which is a machine program.
Each method $f_n$ now has a concrete input/output type which we denote
by $I_n^{\bar{v}}$ and $O_n^{\bar{v}}$ respectively.
$Q[\bar{v}]$ is thus a machine program of type 
$(N, \{I_n^{\bar{v}}\}_{n\in N}, \{O_n^{\bar{v}}\}_{n\in N})$,
and has a set of states that we denote by $S^{\bar{v}}$.
We define the \emph{state space} of $Q[\bar{V}]$, denoted $S^Q$, to be
$\bigcup_{\bar{v}}S^{\bar{v}}$.

Fig.~\ref{fig:program-parametric}(a) shows an example parametric
machine program $Q[\size,T]$, representing a parametric version of the set
program in Fig.~\ref{fig:program}(a).
Given a value 4 for $\size$ and a list $[1,2,3,0]$ for $T$, we get the
machine program $Q[4,[1,2,3,0]]$, which behaves similar to the one of 
Fig.~\ref{fig:program}(a).
We note that the input type of the methods $\mathit{add}$ and
$\mathit{elem}$ depend on the value of the parameter $\size$.

Given two parametric machine programs $Q[\bar{V}]$ and $B[\bar{U}]$ of
type $N$, we are interested in exhibiting a refinement relation
between instances of $Q[\bar{V}]$ and $B[\bar{U}]$.
Let $R$ be a relation on parameter values $\bar{u}$ for $\bar{U}$ and
$\bar{v}$ for $\bar{V}$, given by a predicate on the variables in
$\bar{U}$ and $\bar{V}$.
We say that $Q[\bar{V}]$ \emph{parametrically refines} $B[\bar{U}]$
w.r.t.\@ the condition $R$, if whenever two parameter values
$\bar{u}$ for $\bar{U}$ and
$\bar{v}$ for $\bar{V}$ are such that $R(\bar{u}, \bar{v})$ holds, then 
$Q[\bar{v}]$ refines $B[\bar{u}]$.

We propose a way to exhibit such a conditional refinement, using a
\emph{single} ``universal'' gluing relation, as follows.
Let $Q[\bar{V}]$, $B[\bar{U}]$, and $R$ be as above.
Let $\pi$ be a relation on the state spaces $S^Q$ of $Q[\bar{V}]$ and $S^B$ of
$B[\bar{U}]$, given by a predicate on the variables of $Q[\bar{V}]$
and $B[\bar{U}]$.
We call $\pi$ a \emph{parametric gluing relation} on $Q[\bar{V}]$
and $B[\bar{U}]$.
We say $\pi$ is \emph{adequate}, with respect to the condition
$R$, if the following conditions are satisfied.
In the conditions below, we use the standard Hoare triple notation for total
correctness $\{G\}\ \boxed{P}\ \{H\}$, to mean that a program $P$, when started in a state
satisfying predicate $G$, always terminates in a state satisfying $H$.
We use the superscript $Q$ or $B$ to differentiate the components
pertaining to the programs $Q[\bar{V}]$ and $B[\bar{U}]$ respectively.
\begin{enumerate}
\item (type) For each $n \in N$: $R(\bar{u}, \bar{v}) \implies
  (I_n^{Q,\bar{v}} = I_n^{B,\bar{u}} \myAnd O_n^{Q,\bar{v}} =
  O_n^{B,\bar{u}})$.
\item (init) $\{R\}\, \boxed{\init^{B}();\init^{Q}()}\, \{ \pi \}$.
\item (sim) For each $n \in N$:\ 
 $\{R \myAnd \pi \}\, \boxed{r_B := f^B_n(a); r_Q := f^Q_n(a)}\,  \{ 
	\pi \myAnd r_B = r_Q\}$.
\end{enumerate}
In the program fragments in (init) and (sim) above we assume that
the variable and type declarations are prefixed to the programs.

We can now state the following theorem:
\begin{theorem}
\label{thm:parametric-gluing}
Let $Q[\bar{V}]$ and $B[\bar{U}]$ be parametric machine programs of
type $N$.
Let $R$ be a predicate on $\bar{U}$ and $\bar{V}$,
and let $\pi$ be an adequate parametric gluing relation for $Q[\bar{V}]$ and
$B[\bar{U}]$ w.r.t. $R$.
Then $Q[\bar{V}]$ parametrically refines $B[\bar{U}]$ w.r.t.\@ the
condition $R$.
\end{theorem}

\begin{proof}
Let $Q[\bar{V}]$, $B[\bar{U}]$, and $R$ be as in statement of the
theorem, and let $\pi$ be an adequate gluing relation
w.r.t.\@ $R$.
Consider parameter values $\bar{u}$ and $\bar{v}$ satisfying $R$.
By the (type) condition we know that $B[\bar{u}]$ and $Q[\bar{v}]$
are machines of the same type.
Further, $\pi$ restricted to $S^{\bar{v}} \times S^{\bar{u}}$ is
an adequate gluing relation for the two machines since it
can be seen to satisfy the (init) and (sim) conditions of
Sec.~\ref{sec:adequate}.
\qed
\end{proof}

Consider the parametric machine program
$Q[\size,T]$ in Fig.~\ref{fig:program-parametric}(a), and the
abstract parametric program in 
Fig.~\ref{fig:program-parametric}(b),
which we call $B[\abssize]$.
Consider the condition $R$ which requires that
$\abssize = \size$ and $T$ to be injective.
Let $\pi$ be the parametric gluing predicate
$\forall x: \mathrm{unsigned}, (x < \size) \implies S[T[x]] = absS[x])$.
Then $\pi$ can be seen to be adequate w.r.t.\@ the condition $R$, and
thus $Q[\size,T]$ parametrically refines $B[\abssize]$ w.r.t.\@ $R$.

\paragraph{Verifying Generative Systems using Parametric Refinement.}
Returning to our problem of verifying a generative system, we show a
way to solve this problem using the above framework of conditional
parametric refinement.
Consider a generative system $G$ that given an input specification
$I$, generates a machine program $P_I$, and let
$A_I$ be the abstract specification for input $I$. 
Recall that our aim is to show that for each $I$, $P_I$
refines $A_I$.
Suppose we can associate a parametric program $Q[V]$ with $G$, such
that for each $I$, $G$ can be viewed as generating the value $v_I$ for
the parameter $V$, so that $Q[v_I]$ is behaviourally equivalent to $P_I$.
And similarly a parametric abstract specification $B[U]$, and a
concrete value $u_I$ for each $I$, such that
$A_I$ is equivalent to $B[u_I]$.
Further, suppose we can show that $Q[V]$ parametrically refines $B[U]$
with respect to a condition $R$ on $U$ and $V$.
Then, for each $I$ such that $v_I$ and $u_I$ \emph{satisfy} the condition $R$,
we can conclude that $P_I$ refines $A_I$.
This is illustrated in Fig.~\ref{fig:generative-proof}.
If $R$ is a natural enough condition that a correctly functioning
generator $G$ would satisfy, then this argument would cover all inputs
$I$.

As a final illustration in our running example, to verify the
correctness of the set machine generator $\setgen$, we use the
parametric programs $Q[\size,T]$ and $B[\abssize]$ to capture the
concrete program generated and the abstract specification
respectively.
We then show that $Q[\size,T]$ parametrically refines $B[\abssize]$ w.r.t.\@ the
condition $R$, using the gluing predicate $\pi$, as described above.
We note that the actual values generated for the parameters in this
case (recall that these are values for the  parameters $\size$,
$\abssize$ and $T$) do indeed satisfy the conditions required by $R$,
namely that $\size$ and $\abssize$ be equal, and $T$ be injective.
Thus we can conclude that for each input universe size $n$, the
machine program $P_n$ refines $A_n$, and we are done.

\begin{figure}
\centering
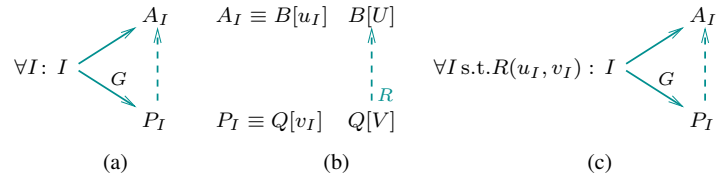
\caption{Proving correctness of a generative system using parametric
  refinement. Part (a) shows the goal, (b) the proof artifacts and
  obligation, and (c) the guarantee. Dashed arrows represent
  refinement, while dashed arrows with $R$ on their tail represent
  parametric refinement w.r.t.\@ $R$.}
\label{fig:generative-proof}
\end{figure}


\section{Intel x86/64 with VMX Support}
\label{sec:processor}

In this section we give a high-level view of the x86/64
processor platform, on which the Muen SK runs.
We abstract some components of the system to simplify our
model of the processor system that the Muen toolchain generates.
For a more complete description of the platform we refer 
the reader 
to the voluminous but 
excellent reference \cite{intel2018}.

The lower part of Fig.~\ref{fig:muen-system} depicts the 
processor system and its components.
The CPU components (with the 64-bit general purpose registers,
including the instruction pointer, stack pointer, and control
registers, as well as model-specific registers like the Time Stamp
Counter (TSC)) and physical memory components are standard.
The layer above the CPUs shows components like the VMCS pointer
(VMPTR), the VMX-Timer, and extended page table pointer (EPTP),
which are part of the VT-x layer of the
Virtual Machine Extension (VMX) mode, that supports virtualization.
The VMPTR component on each CPU points to a VMCS structure, which
is used by the Virtual Machine Monitor (VMM) (here the kernel) to
control the launching and exiting of guest processes or ``subjects.''
The CR3 register (which is technically one of the control registers in the CPU
component) and the EPTP component (set by the active VMCS pointed
to by the VMPTR) control the
virtual-to-physical address translation for instructions that
access memory.
On the top-most layer we show the kernel code (abstracted as a
program) that runs on each CPU. The kernel code has two
components: an ``Init'' component that runs on system initialization,
and a ``Handler'' component that handles VM exits due to interrupts.

We are interested in the VMX mode operation of this processor,
in which the kernel essentially runs as a VMM, and subjects run as
guest software on Virtual Machines 
(VMs) provided by the VMM. Subjects could
be bare-metal application programs, or a guest operating system
(called a VM-subject in Muen).
%
A VM is specified using a VM Control Structure (VMCS), which
stores information about the guest processor state including the
IP, SP, and the CR3 control register values. It also stores values
that control the execution of a subject, like the VMX-timer which sets the
time slice for the subject to run before the timer causes it to
exit the VM, and the
extended page table pointer (EPTP) which translates guest physical
addresses to actual physical addresses. It also contains the
processor state of the host (the kernel).
To launch a subject in a VM, the kernel sets the VMPTR to point to
one of the VMCSs (from an array of VMCSs shown in the figure) using
the VMPTRLD instruction, and then calls VMLAUNCH.
The launch instruction can be thought of as setting the Timer, CR3, and
EPTP components in the VT-x layer, from the VMCS fields.
A subject is caused to exit its VM and return control to the
kernel (called a VM exit), by certain events like VMX-timer
expiry, page table exceptions, and interrupts.

We would like to view such a processor system as a machine of the type
described in Sec.~\ref{sec:classical-refinement}.
The state of the machine is the contents of all its components.
The main operations here are as follows.
\begin{enumerate}
\item \emph{Init}: The init code of the kernel is executed on each of
  the processors, starting with CPU0 which we consider to be the
  bootstrap processor (BSP).
\item \emph{Execute}: This operation takes a CPU id and executes
  the next instruction pointed to by the IP on that CPU. The
  instruction could be one that does not access memory, like an
  instruction that adds the contents of one register into
  another, which only reads and updates the register state on that CPU.
  Or it could be an instruction that accesses
  memory, like moving a value in
  a register $R$ to a memory address $a$. The address $a$ will be
  translated via the page tables pointed to by the CR3 and EPTP components,
  successively, to obtain an address in physical memory which will
  then be updated with the contents of register $R$. Some instructions
  may cause an exception (like an illegal memory access), in which
  case we assume the exception handler runs as part of this
  operation.
  \item \emph{Event}: We consider three kinds of events (or
    interrupts). One is the timer tick event on a CPU. This causes the
    Time-Stamp Counter (TSC) on the CPU to increment, and also
    decrements the VMX-Timer associated with the active VM. If the
    VMX-Timer becomes 0, it causes a VM exit, which is then processed
    by the corresponding handler.
    Other events include those generated by a VMCALL instruction, and
    external interrupts. Both these cause a VM exit. The cause of all
    VM exits is stored in the
    subject's VMCS, which the handler checks and takes appropriate
    action for.
\end{enumerate}

\begin{figure*}
\centering
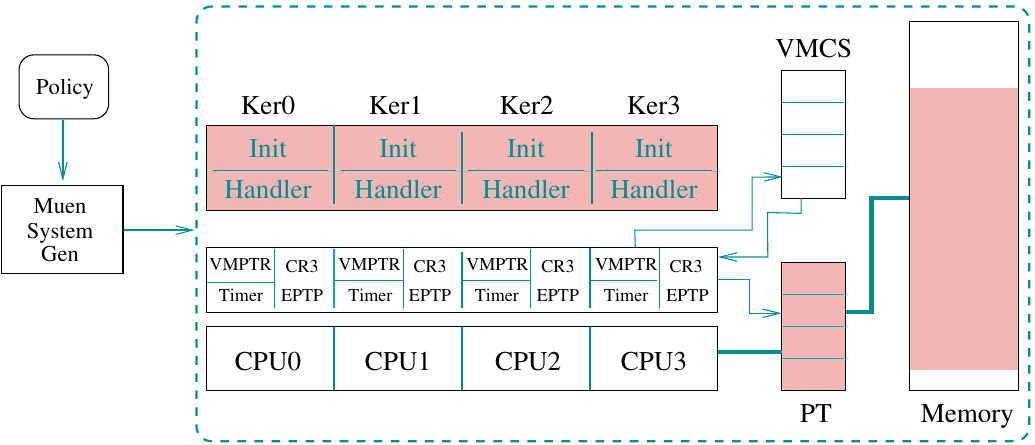
\caption{Components of an x86/64bit Processor System with VMX
  Support. Shaded portions are generated by Muen.}
\label{fig:muen-system}
\end{figure*}




\section{Policy Input Specification}
\label{sec:policy}

The input specification to Muen is an XML file called a \emph{policy}.
It specifies details of the host processor, subjects to be run, and a precise
schedule of execution on each CPU of the host processor system.
%
Details of the host processor include the
number of 
CPUs, the frequency of the processor, and
the available host physical memory regions.
It also specifies for each device (like a keyboard) the
IO Port numbers, the IRQ number, the vector number, and finally
the subject to which the interrupt should be directed.

The policy specifies a set of named \emph{subjects},
and, for each subject, the size and starting addresses of
the components in its virtual memory.
These components include the raw binary image of the subject, 
as well as possible shared
memory ``channels'' it may have declared.
A \emph{channel} is a memory component that can be
shared between two or more subjects.
The policy specifies the size of each channel,
and each subject that uses the channel specifies
the location of the channel in its virtual address space,
along with read/write permissions to the channel.
Fig.~\ref{fig:channel} 
depicts a channel $\mathit{Chan}$
shared by subjects
$\mathit{Sub0}$ (with write permission) and $\mathit{Sub1}$ (with
read permission).

\begin{figure}
	\centering
	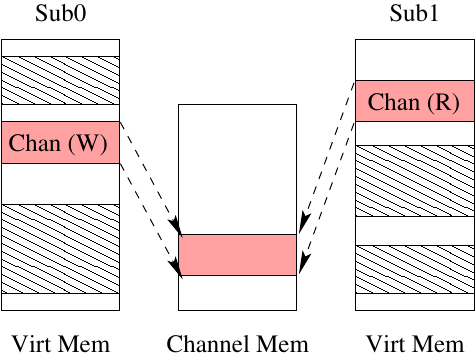
	\caption{A memory channel shared by two subjects}
	\label{fig:channel}
\end{figure}

The schedule is a sequence of \emph{major frames},
that are repeated in a cyclical fashion.
Each major frame specifies a schedule of execution for each CPU,
for a common length of time.
Time is measured in terms of \emph{ticks}, a unit of time
specified in the policy.
A major frame specifies for each CPU a sequence of \emph{minor
  frames}.
Each minor frame specifies a length in ticks, and the subject
to run in this frame.
A subject can be assigned to only one CPU in the
schedule.
An example scheduling policy in XML is shown in
Fig.~\ref{fig:schedpolicy}(a), while
Fig.~\ref{fig:schedpolicy}(b) shows the same schedule viewed as
a clock. Each CPU is associated with one track in the clock.
The shaded portion depicts the passage of time (the tick count) on
each CPU.

\begin{figure}
\begin{minipage}{4.7cm}
\centering
\begin{scriptsize}
\begin{verbatim}
<scheduling tick_rate="10000">
 <major_frame>
  <cpu id="0">
   <minor_fr sub_id="1" ticks="40"/>
   <minor_fr sub_id="2" ticks="40"/>
  </cpu>
  <cpu id="1">
   <minor_fr sub_id="3" ticks="80"/>
  </cpu>
 </major_frame>
 <major_frame>
  <cpu id="0">
   <minor_fr sub_id="1" ticks="80"/>
   <minor_fr sub_id="2" ticks="40"/>
  </cpu>
  <cpu id="1">
   <minor_fr sub_id="4" ticks="60"/>
   <minor_fr sub_id="3" ticks="60"/>
  </cpu>
 </major_frame>
</scheduling>
\end{verbatim}
\end{scriptsize}
\end{minipage}
\begin{minipage}{8cm}
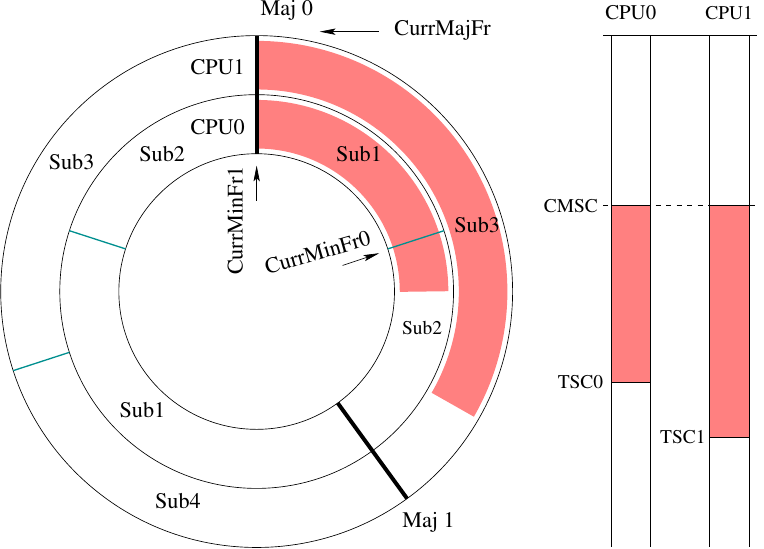
\end{minipage} \\[0.5cm]
\\[0.3cm]
\begin{minipage}[b]{4.7cm}
\begin{center}
\begin{small}
  (a)
\end{small}
\end{center}
\end{minipage}
\begin{minipage}[b]{5.1cm}
\begin{center}
\begin{small}
(b)
\end{small}
\end{center}
\end{minipage}
\begin{minipage}[b]{2cm}
\begin{center}
\begin{small}
\ \ \ \ \ \ \ \ \ \ \ (c)
\end{small}
\end{center}
\end{minipage}
\caption{(a) An Example schedule,
(b) it's clock view, and (c) it's implementation in the Muen kernel.}
\label{fig:schedpolicy}
\end{figure}


Ticks are assumed to occur independently on each CPU, and
can result in a drift between the times on different CPUs.
The scheduling policy requires that before beginning a new frame,
\emph{all} CPUs must complete the previous major frame.
The end of a major frame is thus a synchronization point
for the CPUs.



\section{Muen Kernel Generator}
\label{sec:muen}

Given a policy $C$, the Muen system generator \cite{muen} generates the
components of a processor system $S_C$, which is meant to run
according to the specified schedule. This is depicted in
Fig.~\ref{fig:muen-system}, where the Muen toolchain generates the
shaded components of the processor system, like the initial memory
contents, page tables, and kernel code.
We describe these components in more detail below.

The Muen toolchain first generates a \emph{layout} of the
memory components of all the subjects, in physical memory.
The contents of these components (like the binary file for the binary
component, and the zeroed-out contents of the memory channels) are
filled into an image that will be loaded into physical memory before
execution begins.
%
Based on this layout, the toolchain then generates the page tables for
each subject so that when the subject is running and its page table
is loaded in the CR3/EPTP registers, the virtual addresses will be
correctly mapped to the physical locations.

The Muen toolchain then generates a kernel for each CPU, to orchestrate the
execution of the subjects according to the specified schedule on that CPU.
The kernel is actually a \emph{template} of code written in Spark Ada,
and the Muen toolchain
generates the constants for this template based on the given
policy.
The kernel uses data structures like
\subjectspecs\ to store details like
the page table address and VMCS address for each subject,
and \subjectdescs\ to store the state of 
general purpose registers for a subject.
%
To implement scheduling, the kernel uses the
\schedplans\ structure which is a multidimensional array
representing the schedule for each CPU.
The structure \vectorrouting\ is generated by the toolchain
to represent the table which maps an interrupt vector to 
the corresponding destination subject and the destination vector to be
sent to the destination subject.
The kernel also uses a data structure called \globevents\ for each subject
to save pending interrupts when
the destination subject is not active.
These structures and variables are shown in
Fig.~\ref{fig:kernel_box}. 
\begin{figure}
\centering
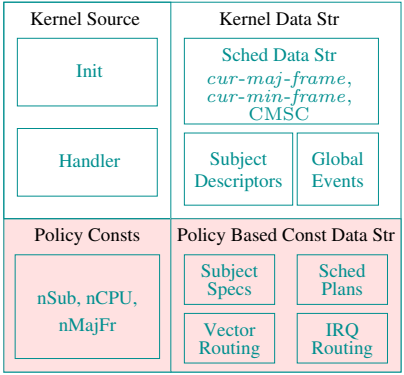
\caption{Components of the generated kernel}
\label{fig:kernel_box}
\end{figure}

The kernel knows the number of ticks elapsed on each CPU from the TSC
register. It uses a shared variable called CMSC (``Current Major
Start Cycle'') to keep track of the 
start of the current major frame.
CMSC is initialized to the value of the TSC on the BSP, at the
time the schedule begins.
Thereafter it is advanced in a fixed periodic manner, based on the
specified length of major frames, whenever a major
frame is completed. This is depicted in
Fig.~\ref{fig:schedpolicy}(c).



We now explain the specific Init and Handler parts of the kernel.
%
%
In the initialization phase the kernel
sets up the VMCS for each subject.
For this the Init
routine reads the 
\subjectspecs\ structure generated by Muen, to fill in
fields like the page table addresses and IP and SP register values, in
each subject's VMCS.
The kernel on each CPU finally sets the VMX Timer value in the VMCS
for the subject that begins the schedule, loads its VMCS address in the
VMPTR, and does a VMLAUNCH.
%

The handler part of the kernel is invoked whenever there is a VM exit.
If the exit is due to a VMX Timer expiry,
the handler checks whether it is at the end of the major frame by
looking at \schedplans. If it is not at the end of
a major frame, it just increments 
the current minor frame.
If it is at the end of a major frame, and not all CPUs have
reached the end of the current major frame, the current CPU waits in a
sense-reversal barrier \cite{herlihy2011art}. 
When all other CPUs reach the end of the
current major frame, they cross the barrier and
the frame pointers are updated by the BSP.
The VMX Timer for the subject to be scheduled next is set
to the time remaining in the current minor frame, calculated using the
fact that (TSC - CMSC) time has already elapsed.
The kernel then does a VMLAUNCH for the subject.

If the exit is due to an external interrupt with some vector $v$,
the VM exit handler finds out the destination subject and the
destination vector corresponding to $v$ from \vectorrouting. 
Then the handler sets the bit corresponding to the destination vector in
\globevents\ for the destination subject.
Whenever the destination subject is ready to handle the interrupt,
the VM exit handler of the kernel \emph{injects} the pending interrupt and 
clears the entry for it in the \globevents\ structure.

In this work we focus on Ver.~0.7 of Muen.
The Muen tool chain is implemented in Ada and C, and comprises about
41K lines of code (LoC).
The kernel template (in Spark Ada) 
is about 3K LoC.

\section{Proof Overview}
\label{sec:proof-overview}

Given a policy $C$, let $S_C$ denote the processor system generated by
Muen.
Let $T_C$ denote an abstract machine spec for the system
$S_C$ (we describe $T_C$ in the next section).
Our aim is to show that for each valid policy $C$, $S_C$
refines $T_C$.
We use the parametric refinement technique of
Sec.~\ref{sec:parametric-refinement} to verify the functional
correctness of the Muen system.
Fig.~\ref{fig:muen-proof} shows how we achieve this.
We first define a parametric program $Q[\bar{V}]$ that models
the generic system generated by Muen, so that for a given policy $C$,
if $\bar{v}_C$ corresponds to the parameter values generated by Muen,
then $S_C$ and $Q[\bar{v}_C]$ are behaviourally equivalent.
In a similar way we define the abstract parametric program
$B[\bar{U}]$, so that with appropriate parameters $\bar{u}_C$,
$B[\bar{u}_C]$ captures the abstract spec $T_C$.
Next we show that $Q[\bar{V}]$ parametrically refines
$B[\bar{U}]$ w.r.t.\@ a condition $R$.
Fig.~\ref{fig:muen-proof} 
shows the proof artifacts and obligations.
Finally, for a given policy $C$, we check that the parameter values
$\bar{u}_C$ and $\bar{v}_C$ satisfy the condition $R$.

In the next few sections we follow this outline to define the
components $B[\bar{U}]$ and $Q[\bar{V}]$, the parametric refinement
between $Q[\bar{V}]$ and $B[\bar{U}]$, and finally the checking of the
condition $R$ for given policy configurations.

\begin{figure}
\centering
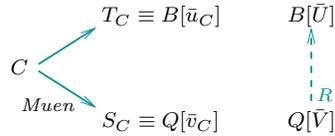
\caption{Muen correctness proof}
\label{fig:muen-proof}
\end{figure}


\section{Abstract Specification}
\label{sec:abstract-spec}

We describe an abstract specification $T_C$
that implements in a simple way the intended behaviour of the system
specified by a policy $C$.
In $T_C$ each subject $s$ is run on a \emph{separate}, dedicated,
single-CPU processor system $M_s$.
The system $M_s$ has its own CPU with registers, and $2^{64}$ bytes of
physical memory $\vmem$.
For each subject $s$ we have a similar sized array called
\perms\ which gives the
permissions (read/write/exec/invalid) for each byte in its virtual
address space.
The policy maps each subject to a CPU of the concrete machine on
which it is meant to run.
To model this we use a set of \emph{logical} CPUs (corresponding
to the number of CPUs specified in the policy),
and we associate with each logical CPU, the
(disjoint) group of subjects mapped to that CPU.
Fig.~\ref{fig:abstract-TC} shows a schematic representation of $T_C$.

\begin{figure}
\begin{center}
  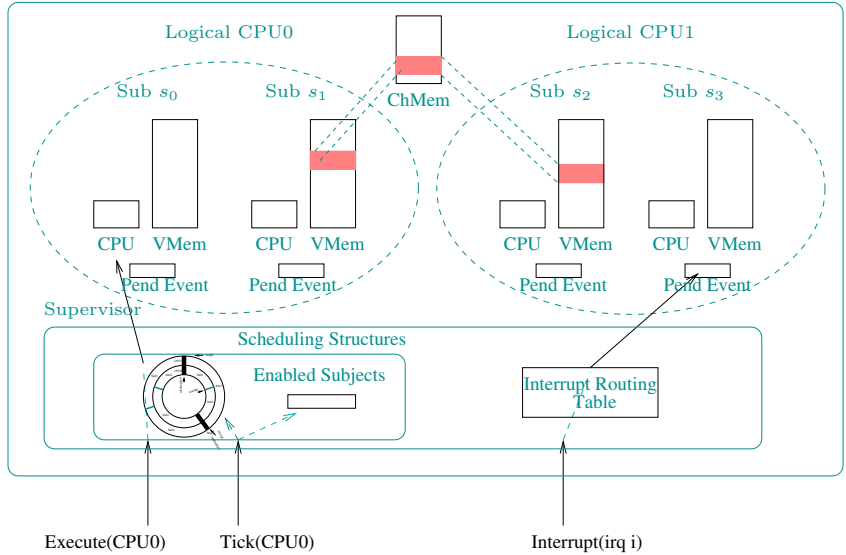
\end{center}
\caption{Schematic diagram of the abstract spcecification $T_C$}
\label{fig:abstract-TC}
\end{figure}

To model shared memory components like channels, we use a separate memory
array $\chmem$, as depicted in 
Fig.~\ref{fig:channel}.
We assume a
partial function $\cmap_s$ which maps the address space of 
subject $s$ to the channel memory.
For any access to
memory, it is first checked whether that location is
in a channel using another Boolean function called $\ch_s$,
and if so, the content is updated/fetched from the address 
in \chmem\ given by the $\cmap_s$ function.

We have used an array of bit-vectors (one for each subject) called
\pendingevents\ to implement interrupt handling.
When bit $i$ of the bit vector for a subject $s$ is set, it represents that 
the interrupt vector $i$ is pending and has to be handled by subject $s$.

There is no kernel in this system, but a \emph{supervisor} whose job 
is to process events directed to a logical CPU or subject, and
to enable and disable subjects based on the scheduling policy and the
current ``time''.
Towards this end it maintains a flag $\enabled_s$ for
each subject $s$, which is set whenever the subject is enabled to run
based on the current time.

To implement the specified schedule, it uses structures \schedplans\,
and \subjectspecs, similar to the Muen kernel.
However it keeps track of time using the clock-like abstraction
depicted in Fig.~\ref{fig:schedpolicy}(b), with the
per-logical-CPU variables \ticks\ (counts each tick event
modulo the total schedule cycle length), 
\minticks\ (reset at the end of every major frame),
\idealmajfp\ (ideal major frame pointer),
\minorfp\ (minor frame pointer),
and \idealcycles\ (cycle counter).
It also uses a global major frame pointer called \majorfp,
and a global cycle counter \cycles.
We use two kinds of pointers here -- ideal and global.
The values of \idealcycles\ and \ticks\ together tell us the
absolute time (number of ticks elapsed) on each logical CPU, and
\idealmajfp\ tells us the major frame based on this time.
The global pointers however are constrained by having to
synchronize on the major frame boundaries, and are updated only
when every CPU has completed the major frame.
We say a CPU is \emph{enabled} if the global and
ideal values of \emph{cycles} match for the CPU, and the global
and ideal values of \emph{majorfp} match.


In the \init\ operation the supervisor initializes the processor
systems $M_s$, permissions array \perms, the channel memory
\chmem, and also the schedule-related 
variables, based on the policy.
The \execute\ operation, given a logical CPU id, executes the next
instruction on the subject machine currently active for that logical
CPU id.
If the CPU is not enabled, the operation is defined to be a no-op.
An \execute\ operation does not affect the state of other subject processors,
except possibly via the shared memory \chmem.
If the instruction accesses an invalid memory address, the system is
assumed to shut down in an error state.
Finally, for the \event\ operation, which is a tick/interrupt event 
directed to a logical CPU or subject, the supervisor updates the
scheduling state, or pending event 
array, appropriately.


To represent the system $T_C$ concretely, we use an Ada program which
we call $A_C$.
$A_C$ is a programmatic realization of $T_C$, with processor registers
represented as 64-bit numeric variables, and memory as byte arrays of
size $2^{64}$.
The operations \init, \execute, and \event\ are implemented as methods
that implement the operations as described above.

Finally, we obtain a parametric program $B[\bar{U}]$ from $A_C$, by
parameterizing it as illustrated in
Sec.~\ref{sec:parametric-refinement}.
Thus we declare constants like MAXSub, MAXCPU, MAXMajfr, and
MAXMinfr to represent the universe of parameter values we allow.
We then declare parameter variables of these size types, like $\ansubs$ (for
``Abstract Number of Subjects'') of type 1..MAXSub,
which will act as parameters to be initialized later.
We call the list of parameters
$\bar{U}$.
We then declare other variables and corresponding arrays of these
sizes (e.g.\@ the array of per-subject CPUs will have size
\ansubs).

By construction, it is evident that if we use the values $\bar{u}_C$ for
the respective parameters in $\bar{U}$, 
we will get a machine program $B[\bar{u}_C]$
which is equivalent in behaviour to $A_C$.

\section{Muen System as a Parametric Program}
\label{sec:muen-parametric}

We now describe how we model the system generated by Muen
as a parametric program.
Let $C$ be a given policy.
To begin with we define a machine program $P_C$ that represents
the processor system $S_C$ generated by Muen.
This is done similar to the abstract
specification $A_C$ in Sec.~\ref{sec:abstract-spec}, except that we now
have a \emph{single} physical memory array which we call $\pmem$.
Further, since the processor system $S_C$ makes use of the VT-x
components, we need
to model these components in $P_C$ as well.
To begin with we represent each page table, represented by an
identifier $\pt$, as a $2^{64}$ size
array $\PT_{\pt}$ of 64-bit numbers, with the translation $\pt(a)$
of an address $a$ being modelled as $\PT_{\pt}[a]$.
We model each VMCS as a structure \vmcs\ with fields 
as defined in \cite{intel2018}.
We use per-CPU variables $\vmptr$ and $\eptp$ that could contain a
VMCS identifier 
or a page table identifier respectively.
We also implement instructions like VMRITE and VMPTRLD as method
calls which read/update these program variables and structures.
We also include all the kernel structures described in
Sec.~\ref{sec:muen} (like \subjectspecs\ and \schedplans) as
global structures in $P_C$.
For each subject $s$, \subjectspecs\ contains a page table id which we
call $\pt_s$.

The operations $\init$, $\execute$, and $\event$ are implemented
as method calls, similar to the abstract spec.
The $\init$ code comes from the Init component of the kernel (see
Sec.~\ref{sec:muen}).
We also initialize the physical memory $\pmem$ with the image
produced by the Muen toolchain.
In the $\execute$ method, memory accesses are translated via the
active page table to access the physical memory $\pmem$.
The implementation of the $\event$ operation comes from
the Handler part of the kernel code (in particular for handling VM
exits due to timer expiry and interrupts).


Finally, we move from $P_C$ to a parametric program $Q[\bar{V}]$.
This is done in a similar way as in
Sec.~\ref{sec:parametric-refinement} and \ref{sec:abstract-spec}.
Apart from the constants like MAXSub,
we use the parameters
\nsubs, \ncpus, \schedplans,
\subjectspecs, 
\vectorrouting,
\pmem, and \PT.
We refer to this list of parameters as $\bar{V}$,
and refer to the resulting parametric program as $Q[\bar{V}]$.
Once again, for an appropriate list of values $\bar{v}_C$
corresponding to a given policy $C$, we believe that $Q[\bar{v}_C]$ is
equivalent to $P_C$, which in turn is equivalent to $S_C$.

\section{Parametric Refinement Proof}
\label{sec:proof}

We now show that the parametric version of the Muen system
$Q[\bar{V}]$ conditionally refines the parametric abstract spec
$B[\bar{U}]$.
From Sec.~\ref{sec:parametric-refinement-theory}, this requires us
to identify the condition $R$, and find a gluing relation
$\pi$ on the state of parametric programs $Q$ and $B$ such that the
refinement conditions (type), (init), and (sim) are satisfied.

We use a condition $R$ whose key conjuncts are the following conditions:
\begin{itemize}
\item $R_1$: The page tables $\pt_s$ associated with a subject $s$
  must be \emph{injective} in that no two virtual addresses, within a
  subject or across subjects, may be mapped to the same physical
  address, \emph{unless} they are specified to be part of a shared
  memory component.
  More precisely, for each $s,s'$, and addresses $a \in \vmem_s$, $a'
  \in \vmem_{s'}$, we have:
  \begin{enumerate}
  \item ($\ch_s(a) \,\myAnd\, \ch_{s'}(a') \,\myAnd\, \cmap_s(a) =
  \cmap_{s'}(a'))$ \\
    $\implies \pt_s(a) = \pt_{s'}(a')$;
  \item $[s \neq s' \myAnd \neg (\ch_s(a) \,\myAnd\, \ch_{s'}(a')
    \,\myAnd\, \cmap_s(a) = \cmap_{s'}(a'))]$ 
    $\implies \pt_s(a) \neq \pt_{s'}(a')$; and
  \item $[a \neq a' \myAnd \neg (\ch_s(a) \,\myAnd\, \ch_{s'}(a')
    \,\myAnd\, \cmap_s(a) = \cmap_{s'}(a'))]$
    $\implies \pt_s(a) \neq \pt_{s'}(a')$.
  \end{enumerate}
\item $R_2$: For each subject $s$, the permissions (rd/wr/ex/present)
  associated with an address $a$ should match with the permissions for
  $a$ in $\pt_s$.
\item $R_3$: For each subject $s$, no invalid virtual address is mapped to a physical address by page table $\pt_s$.
\item $R_4$: For each \emph{valid} address $a$ in a subject $s$
  the
  initial contents of $a$ in $\vmem_s$ should match with that of 
  $\pt_s(a)$ in the physical memory $\pmem$.
\item $R_5$: The values of the parameters (like \nsubs,
  \subjectspecs, \schedplans\ and \IOBitmap) in the concrete should
  match with those in the abstract.
\end{itemize}



The gluing relation $\pi$ has the following key conjuncts:
\begin{enumerate}
\item The CPU register contents of each subject in the abstract
  match with the register contents of the CPU on which the subject
  is active, if the subject is enabled, and with the subject
  descriptor, otherwise.
\item For each subject $s$, and valid address $a$ in its virtual
  address space, the contents of $\vmem_s(a)$ and $\pmem(\pt_s(a))$
  should match.
\item The value $(\mathrm{TSC} - \mathrm{CMSC})$ on each CPU in the
  concrete, should 
  match with how much the ideal clock for the subject's logical CPU is
  ahead of the beginning of the current major frame in the abstract.
\item The major frame pointer in the abstract and concrete should
  coincide, and the minor frame pointers should agree for enabled CPUs.
\item On every enabled CPU, the sum of the VMX Timer and \minticks\
  should equal the deadline of the current minor frame on that CPU.
\item A CPU is waiting in a barrier in the concrete whenever the CPU
  is disabled in the abstract and vice-versa.
\item The contents of the abstract and concrete pending event tables 
should agree.
\end{enumerate}

We carry out the adequacy check for $\pi$, described in
Sec.~\ref{sec:parametric-refinement-theory},
by constructing a ``combined'' version of
$Q$ and $B$ that has the disjoint union of their state variables, as
well as a joint version of their operations, and phrase the
adequacy conditions as pre/post conditions on the joint operations.
In particular, for the (init) condition we need to show that assuming
the condition $R$ holds on the parameters, if we carry out the joint
$init$ operation, then the resulting state satisfies the gluing relation
$\pi$.
To check the (sim) condition for an operation $n$, we assume a joint
state in which the gluing relation $\pi$ holds, and then show that the state
resulting after performing the joint implementation of $n$, satisfies
$\pi$. Once again this is done assuming that the parameters satisfy
the condition $R$.
We carry out these checks using the Spark Ada tool \cite{gnatpro}
which given an Ada
program annotated with pre/post assertions, generates verification
conditions and checks them using provers Z3 \cite{Z3}, CVC4
\cite{BarrettCDHJKRT11}, and Alt-Ergo \cite{conchon12}.

We faced several challenges in carrying out this proof to completion.
A basic requirement of the gluing relation $\pi$ is that the abstract
and physical memory 
contents coincide via the page table map for each
subject.
After a write instruction, we need to argue that this property
continues to hold.
However, even if one were to reason about a given concrete page
table, the prover would never be able to handle this due to the sheer
size of the page table.
The key observation we needed was that the actual mapping of the page table was
irrelevant: all one needs is that the mapping be \emph{injective} as
in condition $R_1$.
With this assumption the proof goes through easily.
A second challenge was proving the correctness of the way the kernel
handles the tick event.
This is a complex task that the kernel carries out, requiring 8
subcases to break up the reasoning into manageable subgoals for both
the engineer and the prover.
The presence of the barrier synchronization before the beginning of a
new major frame, adds to the complexity.
The use of auxiliary variables (like an array of \last\ CPUs), and 
the case-split helped us to carry out this proof
successfully.

Finally, modelling the interrupt handling system, including
injection of interrupts in 
a virtualized setting, was also a complex task.
Here too we had to split the proof of the correctness of interrupt
event into multiple subcases (5 in this case).

Table~\ref{table:proof-result} shows details of our proof effort
in terms of lines of code (LoC) and lines of annotations (LoA) in the
combined proof artifact.
In the combined artifact the LoC count includes 
comments and repetition of code due to case-splits.
The whole proof effort took us an estimated 3-4 person-years.
All proof artifacts used in this project are available at
\url{https://bitbucket.org/muenverification/}.

\begin{table}
\centering
	\begin{tabular}{|c|c|c|c|c|c|}
		\hline
		\multicolumn{2}{|c|}{$B[\bar{U}]$} & \multicolumn{2}{|c|}{$Q[\bar{V}]$} & \multicolumn{2}{|c|}{Combined}\\
		\hline
		LoC & LoA & LoC & LoA & LoC & LoA\\ 
		\hline
		793 & 0 & 1,914 & 0 & 13,970 & 6,214\\
		\hline
	\end{tabular}
\vspace{0.2cm}
\caption{Proof artifact sizes}
\label{table:proof-result}
\end{table}




\section{Checking Condition $R$}
\label{sec:condition-checking}
We now describe how to efficiently check that for a given
policy $C$, the parameters generated by Muen and those of the abstract 
specification, satisfy the condition $R$.
%
Let $N_s$ be the number of subjects.
Let $N_v$ be the maximum size of the virtual address
space, and $N_u$ the actual \emph{used} space 
which is the total size of the virtual memory components of a subject, for a subject.
$N_v$ is typically of the order $2^{48}$ bytes (approx 256TB) while $N_u$ is of the
order of 1.5GB.
Let $P$ denote the size of a paging structure in bytes, which is
of the order of 500K.
A naive way to check the conditions $R_1$ (injectivity), $R_2$ (permission-checking), 
$R_3$ (validity), and $R_4$ (content-matching), would be to use a bit
array for physical memory and iterate over \emph{all} virtual
addresses to check these conditions.
However, this runs in time proportional to $N_v$, and such an
algorithm would take days to run.
In contrast, we give a way to check these properties in time proportional to
$N_u$ for each subject. These algorithms run in a few seconds.

To check these properties efficiently we exploit the following facts
\begin{itemize}
	\item A policy specifies memory in terms of memory components 
				which are blocks of contiguous pages.
	\item Muen generates an intermediate representation of
				physical memory (called the ``B-policy'') 
				which contains the list of all memory components 
				which have to be placed in the physical memory.
				We call this list PMemComponents. 
				Each physical memory component is specified 
				with a name called ``physical name'', size, and 
				the physical address where it has to be placed in the physical memory.
				The B-policy also contains the list of virtual memory components 
				for each subject which actually are mapped to physical memory components.
				The list of virtual memory components for a subject $s$ 
				is denoted by VMemComponents($s$).
				With each virtual memory component in a subject, 
				its name called ``logical name'', 
				logical address which is the address in the subject's virtual memory,
				the permissions of the subject,
				and the physical name of the physical memory component 
				to which it is mapped to are specified.
	\item The sum of size of all the memory components is usually much less than the physical address space
        i.e. $N_u$ is small in comparison to $N_v$.
\end{itemize}

\paragraph{Injectivity-check $R_1$} 
Checking the injectivity property of the page tables 
naively will require us to check 
for all pairs of distinct logical addresses 
possibly in two different subjects' address spaces
whether they are mapped to the same physical address 
via the corresponding page tables.
The running time of this naive check will be proportional to $N_v$
where $N_v$ is approximately 256T.

Since the memory components are blocks of contiguous pages and contiguous frames 
in virtual and physical address space respectively, 
checking the injectivity property reduces to checking 
whether any two memory components across subjects overlap (actually coincide) in physical memory.
We first check that the memory components in PMemComponents have been laid out
disjointly in the physical memory layout given in the B-policy, by checking for overlap.
Next we check whether the physical memory components 
corresponding to any pair of virtual memory components 
possibly from different subjects are overlapping.
If such an overlap is found, to flag it as an illegal sharing 
we also need to check whether 
one of them is not specified as a channel in the original policy and
one of them is writable. We check that one of them is writable because 
Muen allows sharing of read-only components to reduce the physical memory footprint.
Finally we check whether the page tables are generated according to the
layout in the B-policy. 
For this purpose we translate the virtual address of each page 
in every virtual memory component of each subject and 
check that it matches with the physical address calculated via the B-policy.
This algorithm assumes that $R_3$ holds (validity-check passes) because 
we have not checked this property for invalid addresses 
(logical addresses outside the memory components)
in subjects' virtual address spaces.

\begin{algorithm}[t]
	\begin{algorithmic}[1]
		\ForAll{$A$, $B$ in PMemComponents}
			\If{$ A.\PA  \leq B.\PA $}
				\If{$ A.\PA + A.Size  \geq \ B.\PA $}
					\State Output(\enquote{Illegal sharing detected.});
				\EndIf
			\EndIf	
		\EndFor
	\end{algorithmic}
	\caption{Algorithm to check overlap of physical memory components}
	\label{alg:overlap-physical}
\end{algorithm}

\begin{algorithm}[t]
	\begin{algorithmic}[1]
		\State $\mathcal{M} = \bigcup_{subject\ s} $VMemComponents($s$)
		\ForAll{$A$, $B$ in $\mathcal{M}$}
			\If{$A.\PA \leq B.\PA$}
				\If{$ A.\PA + A.Size  \geq \ B.\PA $}
					\If{$\neg\ (isChannel(A) \land isChannel(B))$} 
						\If{$isWritable(A) \lor isWritable(B)$} 
							\State Output(\enquote{Illegal sharing detected.});
						\EndIf
					\EndIf
				\EndIf
			\EndIf
		\EndFor
	\end{algorithmic}
	\caption{Algorithm to check overlap of virtual memory components in physical memory}
	\label{alg:overlap-virtual}
\end{algorithm}

The algorithm to check the overlapping of
two physical memory components in the physical memory layout
is shown in Algorithm~\ref{alg:overlap-physical}.
The algorithm to check the overlapping of physical memory components 
corresponding to any two virtual memory components 
is shown in Algorithm~\ref{alg:overlap-virtual}.
Finally the algorithm to check whether page tables are generated 
according to the layout in the B-policy is shown in Algorithm~\ref{alg:pt-check}.
$A.\LA$ and $A.\PA$ denote the logical address of the memory component $A$ and 
the physical address corresponding to
the logical address of the memory component $A$ respectively.
$A{.}Size$ is the size of the memory component $A$ as given in the B-policy.
In Algorithm~\ref{alg:pt-check} $p.\LA$ and $p.\PA$
denote the logical and physical address of a page $p$ respectively 
as mentioned in the B-policy.

\begin{algorithm}[t]
	\begin{algorithmic}[1]
		\ForAll{$s$ in subjects}
			\ForAll{$A$ in VMemComponents($s$)}
				\State $\page.\LA$ = $A.\LA$; \ \ // logical address
				\State $\page . \PA$ = $A . \PA$; \ \ // physical address
				\While{$\page . \LA < A . \LA + A . \psize$}
					\If{$\pt_s(\page . \LA) \neq \page . \PA$}
						\State Output(\enquote{Address mismatch});
					\EndIf
					\State $\page . \LA = \page . \LA + page\_size$;
					\State $\page . \PA = \page . \PA + page\_size$;
				\EndWhile
			\EndFor
		\EndFor
	\end{algorithmic}
	\caption{Algorithm to check whether page tables are generated according to B-policy}
	\label{alg:pt-check}
\end{algorithm}

\paragraph{Permission-check $R_2$} 
Again a naive algorithm to check $R_2$ will check permission for each page
in the corresponding paging structure for each subject and 
match it with the permission given in the policy.
The running time of this algorithm will be of the order $N_sN_v$.
To check $R_2$ efficiently, 
for each virtual memory component in every subject
we check that the permission of the pages in the virtual memory component 
are set according to the permissions mentioned in the B-policy.
Again this check is dependent on the validity-check as 
we do not check the permissions for invalid addresses.


\paragraph{Validity-check $R_3$}
To check that only valid addresses are mapped by the page tables,
a naive way will be to translate each virtual address in each subject and 
check that only for valid virtual addresses
present bit is set in the paging structures.
The running time of this algorithm will be proportional to $N_v$ 
which is huge and it will take days to run.

Muen generates a file per subject to store the page tables based on the B-policy.
Paging structures for a subject
(PT, PD, PDPT and PML4)
are stored in this file as a contiguous sequence of 64-bit words.
%
%
In this file, 
in the beginning the entries of the PML4 paging structure are stored 
which are followed by the entries of the PDPT paging structures.
Then the entries of the PDPT paging structures are followed by 
all the entries of the PD paging structures and 
finally the entries of the PT paging structures are stored.

To check $R_3$ efficiently, we observe that the translation of a valid virtual
address makes use of certain entries of the paging structures
in the file containing the paging structures. 
The present-bits in \emph{only} these entries should be set, and all
others should be unset.
To check this we use arrays $\PTBitArray_s$ of bits, one for each
64-bit entry in the file containing the paging structures for the subject $s$.
We translate each valid virtual address, and set the bits in the
array $\PTBitArray_s$ that correspond to each paging structure entry accessed in
the translation.
After translating all valid addresses of the subject,
we check that there is no entry in the paging structure with the
present-bit set but associated array bit unset. This is
shown in Algorithm~\ref{alg:validity}.
In Algorithm~\ref{alg:validity} $PTFile_s$ is the file containing paging structures for subject $s$.
The function Entries($PTFile_s$) returns the set of entries available in the file $PTFile_s$.
The function TestPresentBit($i$) checks whether present bit is set in the entry $i$.
We note that this algorithm runs in time $O(N_u+P)$ for each subject.

\begin{algorithm}
	\begin{algorithmic}[1]
		\State Create a Boolean array $\PTBitArray_s$ for each subject $s$
		\ForAll{$s$ in subjects}
			\ForAll{$A$ in VMemComponents($s$)}
				\State //Marking the bits in $\PTBitArray_s$ corresponding to the entries used for translation of valid virtual addresses
				\ForAll {pages $p$ in $A$}
					\State   setBit($\PTBitArray_s$, PML4\_offset) 
					\State   setBit($\PTBitArray_s$, PDPTE\_offset)
					\State   setBit($\PTBitArray_s$, PDE\_offset)
					\State   setBit($\PTBitArray_s$, PTE\_offset)
				\EndFor
			\EndFor
			\State //Checking present bit is set in any other entries
			\ForAll {$i$ in Entries($PTFile_s$)}
				\If{(TestPresentBit($i$) $= 1$) $\land$  ($\PTBitArray_s[i]$ = 0)}
					\State Output(\enquote{Invalid Page Table Entry $i$})
				\EndIf
			\EndFor
		\EndFor
	\end{algorithmic}
	\caption{Validity checking algorithm}
	\label{alg:validity}
\end{algorithm}

\paragraph{Content-matching $R_4$ and Structure-matching $R_5$} 
Condition $R_4$ (initial memory contents) is straightforward to
check.
We simply compare the content of the image generated by Muen with
each individual component's content (which is a specified file or
``fill'' element) byte by byte.
%
%
Similarly, condition $R_5$ is checked by algorithmically generating
the abstract parameter data structures and ensuring that the Muen
generated ones conform to them.

\paragraph{Results.}
We implemented our algorithms above in C and Ada, using the Libxml2
library to process policy files, and
a Linux utility \texttt{xxd} to convert the
Muen image and individual files from raw format to
hexadecimal format.

We ran our tool on 16 system configs, 9 of which (D7-*,D9-*)
were available
as demo configurations from Muen.
The remaining configs (DL-*) were configured by us to mimic
a Multi-Level Security (MLS) system 
from \cite{rushbysep}.
Details of representative configs are shown 
in Table~\ref{tab:table_config}.
For each configuration the table columns show the number of subjects,
number of CPUs, the size of physical memory needed on the processor,
the ISO image size generated by the Muen toolchain, the time taken by our condition
checking tool, and finally whether the check passed or not.

We used the 3 configs D9-* (from Ver.~0.9 of Muen) as seeded
faults to test our tool. Ver.~0.9 of Muen generates \emph{implicit}
shared memory components, and this undeclared sharing was
correctly flagged by our tool.

The average running time on a configuration was 5.6s.
The experiments were carried out on
an Intel Core i5 machine with 4GB RAM running Ubuntu 16.04.

\begin{table}
\begin{center}
\begin{small}
\begin{tabular}{|l|r|r|r|r|r|c|}
    \hline
    System
    & Sub
    & CPU
    & PMem
    & Image
    & Time 
    & Check   
    \\
     & & & (MB) & (MB) & (s) & Passed \\
    \hline
        D7\_Bochs  & 8  & 4 & 527.4  & 13.8 & 3.7  & \checkmark\\
    \hline
        DL\_conf1  & 8  & 4 & 506.5 & 12.9 & 3.7 & \checkmark\\
        DL\_conf2 & 9  & 4 & 1552.7 & 15.1& 6.8 & \checkmark\\        
        DL\_conf3  & 12  & 4 & 1050.1  & 23.3 & 6.7  & \checkmark\\
        DL\_conf4  & 16  & 4 & 1571.4  & 15.1 & 9.2  & \checkmark\\
    \hline
        D9\_Bochs  & 10 & 2 & 532.9  & 16.2 & 4.9 & \ding{55}\\
        D9\_vtd     & 16 & 4 & 1057.8 & 18.4 & 5.9 & \ding{55}\\
        D9\_IntelNuc & 10 & 2 & 567.0 & 16.2 & 5.5 & \ding{55}\\         
     \hline
\end{tabular}
\end{small}
\end{center}
\caption{Checking condition $R$ on different configs}
\label{tab:table_config}
\end{table}

\section{Basic Security Properties}
\label{sec:security-properties}

While we believe that the property we have proved for Muen in this paper
(namely conformance to
an abstract specification via a refinement proof)
is the canonical security property needed of a
separation kernel, security standards often require
some specific basic security properties to be satisfied.
We discuss below how some of these properties mentioned in
\cite{SKPP,heitmeyer} follow from the verification exercise we
have carried out for Muen.

Let us consider a system $S_C$ generated by Muen,
and the abstract specification $T_C$, for a given
policy $C$. Let us further assume that the generated parameters
satisfy the condition check of Step~2.
Then we know that every (error-free) sequence of operations in the
concrete system $S_C$ can be simulated, via the gluing relation
$\rho$ built on injective page tables, by the abstract system $T_C$.
We now argue that the system $S_C$ must satisfy the specific
security properties below.

\textbf{No exfiltration} This property states that 
the execution of a subject does not influence the state of another
subject.
In our setting, we take this to mean that a write to a memory
address $a$ by a subject $s$ does not affect the memory contents
of subject $s'$, assuming address $a$ is not part of a declared
memory channel between $s$ and $s'$.
Let us consider a sequence of operations in $S_c$ leading to state
$S_i$ where subject $s$ performs an exfiltrating write to address
$a$.
Now by our refinement proof, the abstract system $T_C$ must be
able to simulate the same initial sequence of operations and reach
a state $T_i$ which is glued to $S_i$ via the gluing relation
$\rho$, which in turn is based on page table maps satisfying the
injective property $R_1$.
Now when $s$ makes an exfiltrating write to address $a$ to take
the concrete state from $S_i$ to $S_{i+1}$, the
memory contents of $s'$ must change in going from state $S_i$
to $S_{i+1}$.
However when we perform the same write in the abstract state
$T_i$, the memory contents of $s'$ do not change in going from
$T_i$ to $T_{i+1}$. It follows that the state $T_{i+1}$
\emph{cannot} be glued via $\rho$ to the concrete state
$S_{i+1}$.
This contradicts the simulation property of our proof.
Thus, it follows that no subject can perform an exfiltrating
write.
A similar argument holds to show that $s$ cannot change the other
components (like the register contents) of the state of $s'$.


\textbf{No infiltration} This property states that 
an operation by a subject $s$ should not be influenced by the
state of another subject $s'$.
More precisely, suppose we have two concrete states $S_1$ and
$S_2$ of $S_C$ in which the state of subject $s$ is identical. In
our setting this means that $S_1$ and $S_2$ are glued,
respectively, to abstract states $T_1$ and $T_2$ in which the
states of $s$ are identical.
Now suppose subject $s$ performs an operation (say a read of a
memory location) in $S_1$ and $S_2$ to reach $S_1'$ and $S_2'$
respectively.
Then the state of $s$ in $S_1'$ and $S_2'$ should be identical.
However, this follows from our proof, since by construction the
state of $s$ in $T_1'$ and $T_2'$ obtained by performing the same
operation in the abstract states $T_1$ and $T_2$ respectively,
\emph{must be} identical.
Since $S_1'$ and $S_2'$ must be glued to $T_1'$ and $T_2'$
respectively, the property follows.

\textbf{Temporal separation} This property states that subjects
are executed according to the specified schedule, and that while
they are inactive their state does not change. The latter could
happen for instance if the register state of the previously
executing subject was exposed by not restoring the current
subject's state correctly.
Once again this property follows in our setting since every
sequence of operations by $S_C$ must be matched by the abstract
specification, and by construction the abstract specification
executes according to the specified policy and the state of a
subject does not change while inactive.

We note that the property of non-bypassability from \cite{SKPP}
would require the above three properties to hold.

\textbf{Kernel integrity} This property states that
the kernel state, including its code and data, is not affected
by the operations carried out by a subject.
This property is called the tamper-proof property in \cite{SKPP}.
Though this property is not directly modelled in our setting (note
that we model the kernel code and data as a high-level program
that cannot be accessed by subjects), while checking
condition $R$ we also check that the page tables
generated by Muen satisfy the injective property across all memory
components, including the kernel components, as specified in the
B-policy of Muen.
This effectively ensures the integrity of the kernel.


\section{Discussion}
\label{sec:discussion}

The validity of the verification proof carried out in this work
depends on several assumptions we have
made.
Some implicit assumptions we have made
include the fact that
processor hardware components like page table translation and
VMX instructions behave the way we have modelled
them. 

In addition there are several explicit assumptions related to the
way
we have modelled the abstract specification of how
the SK is expected to behave:
\begin{itemize}
\item
When scheduling actually
begins after the initialization routines, the TSCs on all CPUs have the
\emph{same} value.
\item
If any subject performs an illegal instruction (like
  accessing an invalid memory address) the system halts in an
  error state.
\item
The tick count on the 64-bit TSC counter does not overflow
  (this is a mild assumption as it would take
  \emph{years} to happen); Similarly we
  assume that a minor frame length is never more than $2^{32}$
  ticks as the VMX Timer field is only 32-bits wide.
\end{itemize}
If any of these assumptions are violated, the proof will not go
through, and in fact we would have counter-examples to
conformance with the abstract specification.

\begin{figure*}[ht]
\begin{center}
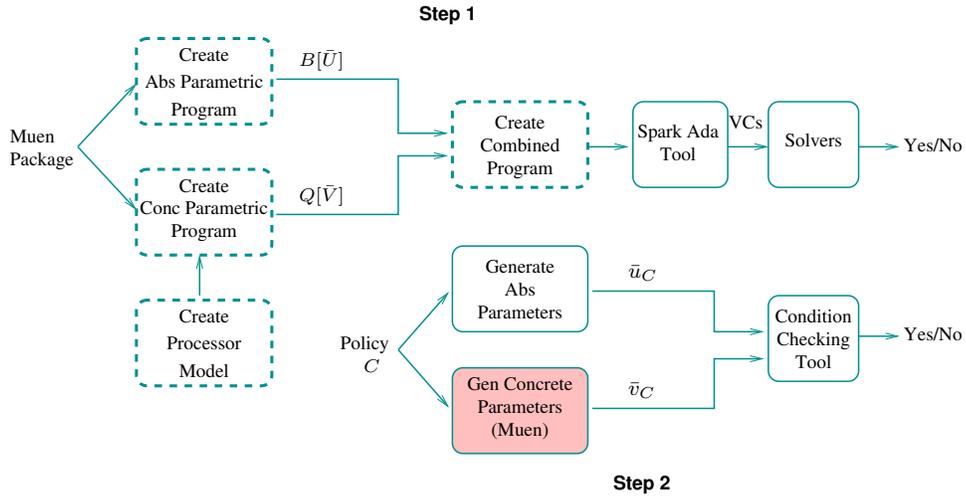
\end{center}
\caption{Components in Muen verification. Untrusted components are
  shown shaded while non-automated (manual) steps are shown with
  dashed boxes.}
\label{fig:components-trust}
\end{figure*}

Finally, we show the various components used in our verification in
Fig.~\ref{fig:components-trust}.
Each box represents a automated tool (full boxes) or manual
transformation carried out (dashed boxes).
Components that we trust in the proof are unshaded, while untrusted
components are shown shaded.

We would like to mention that the developers of Muen were
interested in adding our condition checking tool to the Muen
distribution, as they felt it would strengthen the checks they
carry out during the kernel generation. We have updated
our tool to work on the latest version (v0.9) of Muen, and handed
it over to the developers.

\section{Related Work}
\label{sec:related}

We classify related work based on general OS verification,
verification of separation kernels, and translation validation based
techniques.

\paragraph{Operating System Verification.} 
There has been a great deal of work in formal verification of
operating system kernels in the last few decades.
Klein \cite{Klein3} gives an excellent survey of the
work till around 2000.
In recent years the most comprehensive work on OS verification has
been the work on seL4 \cite{Klein2}, which gave a
refinement-based proof of the functional
correctness of a microkernel using the Isabelle/HOL theorem prover.
They also carry out an impressive verification of page table
translation \cite{TuchKleinOSV04}.	
The CertiKOS project \cite{certikos} provides a technique for
proving contextual functional correctness across the
implementation stack of a kernel, and also handles concurrency.
Other recent efforts include verification of a type-safe
OS \cite{YangH10}, security invariants in ExpressOS
\cite{expressos}, 
and the Hyperkernel project \cite{Luke}.
%

While verification of a general purpose OS is a more complex task
than ours---in particular a general kernel has to deal with dynamic
creation of processes while in our setting we have a \emph{fixed} set of
processes and a fixed schedule---the techniques used there cannot
readily reason about generative kernels like Muen.
We would also like to note here that while it is true in such
verification one often needs to reason about parametric components
(like a method that computes based on its parameters), the whole
programs themselves are \emph{not} parametric.
In particular, a standard operating system is \emph{not} parametric: it
begins with a concrete initial state, unlike a parametric program in
which the initial state has unitialized parameters.
Thus the techniques developed in this paper are needed
to reason about such programs.


Finally, we point out that none of these works address the use of VT-x
virtualization support.



\paragraph{Verification of Separation Kernels.}
There has been substantial work in formal verification of
separation kernels/hypervisors. seL4 \cite{Klein2} can
also be configured as a separation kernel,
and the underlying proof of functional correctness
was used to prove information flow enforcement.
Heitmeyer et al \cite{heitmeyer}
proved data separation properties 
using a refinement-based approach for a special-purpose SK called
ED, in an embedded setting.
As far as we can make out these systems are not generative in nature, 
and either do not use or do not verify hardware virtualization support.
Additionally, unlike our work, none of these works (including OS
verification works) are \emph{post-facto}: they are developed 
\emph{along} with verification.

Dam et al \cite{Dam} verify a prototype SK called PROSPER,
proving information flow security on the specification and showing a
bisimulation between the specification and the implementation. 
PROSPER works for a minimal configuration with exactly two subjects,
and is not a generative system. 
%
The Verisoft XT project \cite{verisoftxt} 
attempted to prove the correctness of Microsoft's Hyper-V hypervisor
\cite{Leinenbach:FM2009-80}
and Sysgo's PikeOS, using VCC \cite{VCC}.
While the Hyper-V project was not completed, 
the PikeOS memory manager was proved correct in \cite{Baumann}.
Sanan et al \cite{sanan1} propose an approach towards
verification of the XtratuM kernel \cite{crespo} in Isabelle/HOL, but the
verification was not completed.

\paragraph{Translation Validation Techniques.}
Our verification problem can also be viewed as translation validation
problem, where the Muen generator translates the input policy
specification to an SK system.
The two kinds of approaches here aim to verify the generator code
itself (for example the CompCert project \cite{xavier}) which can be
a challenging task in our much less structured, \emph{post-facto}
setting; or aim to verify the generated output for each specific 
instance \cite{pnueli}.
Our work can be viewed as a via-media between these two approaches:
we leverage the template-based nature of the generated system to
verify the generator conditionally, and then check whether the
generated parameter values satisfy our assumed conditions.





\section{Conclusion}
\label{sec:conclusion}

In this work we have proposed a technique to reason about
\emph{template}-based generative systems,
and used it to carry out effective \emph{post-facto} verification of
the separation property of a
complex, generative, virtualization-based separation kernel.
In future work we plan to extend the scope of verification to
address concurrency issues that we presently ignore in this
work.

\emph{Acknowledgement.}
We thank the developers of Muen,
Reto Buerki and Adrian-Ken Rueegsegger, 
for their painstaking efforts in helping us understand the Muen
separation kernel.
We also thank Arka Ghosh
for his help in the proof of interrupt handling.

%
%
%
\newpage
\bibliographystyle{splncs04}
\bibliography{references}
%




\end{document}